\documentclass[twoside,11pt]{article} 
\usepackage{jmlr2e}
\usepackage{times,graphicx,times,amsfonts,amsthm,amsmath,amscd,amssymb,algorithm,algorithmic}

\allowdisplaybreaks[0]

\newtheorem{definition}{Definition}

\newtheorem{theorem}{Theorem}
\newtheorem{lemma}{Lemma}

\newtheorem{proposition}{Proposition}
\newcommand{\MA}{{the partition scheme}}

\newcommand{\bea}{\begin{align}}
\newcommand{\eea}{\end{align}}

\renewcommand{\Pr}{\mathbb P}

\newcommand{\cA}{\mathcal{A}}

\newcommand{\beq}{\begin{eqnarray}}
\newcommand{\eeq}{\end{eqnarray}}
\newcommand{\beqn}{\begin{equation}}
\newcommand{\eeqn}{\end{equation}}

\newcommand{\Rp}{\mathbb{R}_+}

\newcommand{\lf}{\left}
\newcommand{\rf}{\right}

\newcommand{\Sg}{\Sigma}

\newcommand{\bdg}{\mathbf{d_G}}

\newcommand{\bone}{\mathbf{1}}

\newcommand{\mc}{\mathcal}
\newcommand{\mb}{\mathbf}

\renewcommand{\hat}{\widehat}
\newcommand{\cM}{\mc{M}}

\newcommand{\cW}{\mc{W}}
\newcommand{\cB}{\mc{B}}

\newcommand{\cG}{\mc{G}}
\newcommand{\cH}{\mc{H}}

\newcommand{\bB}{\mathbf{B}}

\newcommand{\bx}{\mb{x}}
\newcommand{\bX}{\mb{X}}

\newcommand{\bQ}{\mathbf{Q}}

\newcommand{\E}{\mathbb{E}}

\ShortHeadings{Partition Merge}{Blondel, Jung, Kohli and Shah}
\firstpageno{1}

\title{Partition-Merge: Distributed Inference and Modularity Optimization}

\author{\name Vincent Blondel \email vincent.blondel@uclouvain.be \\
       \addr Universite Catholique de Louvain, Belgium
       \AND
       Kyomin Jung \email kjung@snu.ac.kr \\
       \addr Seoul National University, South Korea
       \AND
       Pushmeet Kohli \email pkohli@microsoft.com \\
       \addr Microsoft Research, Cambridge, UK
       \AND
       \name Devavrat Shah \email devavrat@mit.edu \\
       \addr MIT, Cambridge, USA}

\editor{Editor}

\begin{document}

\maketitle

\begin{abstract}
This paper presents a novel meta algorithm, Partition-Merge (PM), which takes existing centralized algorithms for graph computation and makes them distributed and faster. In a nutshell, PM divides the graph into small subgraphs using our novel randomized partitioning scheme, runs the centralized algorithm on each partition separately, and then {\em stitches} the resulting solutions to produce a global solution. We demonstrate the efficiency of the PM algorithm on two popular problems: computation of Maximum A Posteriori (MAP) assignment in an arbitrary pairwise Markov Random Field (MRF), and modularity optimization for  community detection. We show that the resulting distributed algorithms for these problems essentially run in time linear in the number of nodes in the graph, and perform as well -- or even better -- than the original centralized algorithm as long as the graph has geometric
structures\footnote{Roughly speaking, a graph has geometric structures, or polynomial growth property, when the number of nodes within
distance $r$ of any given node grows no faster than  a polynomial function of $r$.}. More precisely, if the centralized algorithm is a $\mathcal{C}-$factor approximation with constant $\mathcal{C}\ge 1$, the resulting distributed algorithm is a $(\mathcal{C}+\delta)$-factor approximation for any small $\delta>0$; but if
the centralized algorithm is a non-constant (e.g. logarithmic) factor approximation, then the resulting distributed algorithm becomes a constant factor approximation. For general graphs, we  compute explicit bounds on the loss of  performance of the resulting distributed algorithm with respect to the centralized algorithm.
\end{abstract}


\section{Introduction}
Graphical representation for data has become central to modern large-scale data processing applications. For many of these applications, large-scale data computation boils down to solving problems defined over massive graphs. While the theory of centralized algorithms for graph problems is getting reasonably well developed, their distributed (as well as parallel) counterparts are still poorly understood and remain very active areas of
current investigation. Moreover, the emerging cloud-computation architecture is making the need of distributed solutions even more pressing.

\medskip
\noindent{\bf Summary of results.} In this paper, we take an important step towards this challenge. Specifically, we present a meta algorithm, Partition-Merge (PM), that makes existing centralized (exact or approximate) algorithms for graph computation distributed and faster without loss of performance, and in some cases, even improving performance. The PM meta algorithm is based on our novel  partitioning of the graph
into small disjoint subgraphs.  In a nutshell, PM partitions the graph into small subgraphs, runs the centralized algorithm on each partition separately (which can be done in distributed or parallel manner); and finally {\em stitches} the resulting solutions to produce a global solution.
We apply the PM algorithm to two representative class of problems: the MAP computation in a pairwise MRF and modularity optimization
based graph clustering.

The paper establishes that for any graph that satisfies the polynomial growth property, the resulting distributed PM based implementation of the original centralized algorithm is a $(\mathcal{C} + \delta)$-approximation algorithm whenever the centralized algorithm is a $\mathcal{C}$-approximation algorithm for some constant $\mathcal{C} \geq 1$. In this expression, $\delta$ is a small number that depends on a tuneable parameter of the algorithm that affects the size of the induced subgraphs in the partition; the larger the subgraph size, the smaller the $\delta$. More generally, if the centralized algorithm is an $\alpha(n)$-approximation (with $\alpha(n) = o(n)$) for a graph of size $n$, the resulting distributed algorithm becomes a constant factor approximation for graphs with geometric structure! The computational complexity of the algorithm scales linearly in $n$. Thus, our meta algorithm can make centralized algorithms, faster, distributed and improve its performance. 

The algorithm applies to any graph structure, but strong guarantees on performance, as stated above, require geometric structure. However, it is indeed possible to explicitly evaluate the loss of performance induced by the distributed implementation compared to the centralized algorithm as stated in Section \ref{sec:gen}.

A cautionary remark is in order. Indeed, by no means, this algorithm means to answer all problems in distributed computation. Specifically, for dense graph, this algorithm is likely to exhibit poor performances and definitely such graph structure would require a very different approach. Our meta algorithm requires that the underlying graph problem is {\em decomposable} or {\em Markovian} in a sense. Not all problems have this structure and these problem therefore require different way to think about them.

%\medskip
\subsection*{Related Work and Our contributions} The results of this paper, on one hand, are related to a long list of works on designing distributed algorithms for decomposable problems. On the other hand, the applications of our method to MAP inference in pairwise MRFs and clustering relate our work to a large number of results in these two respective problems. We will only be able to discuss very closely related work here.

We start with the most closely related work on the use of graph partitioning  for distributed algorithm design. Such an approach is quite old; see, e.g., ~\cite{decomp1, decomp2} and ~\cite{decomp3} for a detailed account of the approach until 2000. More recently, such decompositions have found wide variety of applications including local-property testing \cite{decomp4}. All such decompositions are useful for {\em homogeneous} problems, e.g. for finding maximum-{\em size} matching or independent set rather than the {\em heterogenous} maximum-{\em weight} variants of it. To overcome this
limitation, a different (somewhat stronger) notion of decomposition was introduced by Jung and Shah \cite{nips08} for minor-excluded graphs that built upon \cite{decomp2}. All of these results effectively partition the graph into small subgraphs and then solve the problem inside each small subgraph  using exact (dynamic programming) algorithms. While this results in a $(1+\epsilon)$-approximation algorithm for any $\epsilon > 0$ with
computation scaling essentially linearly in the graph size ($n$), the computation constant depends super-exponentially in $1/\epsilon$. Therefore, even with $\epsilon = 0.1$, the algorithms become unmanageable in practice.

As the main contribution of this paper, we first propose a novel graph decomposition scheme for graphs with geometry or polynomial growth structure.
Then we establish that by utilizing this decomposition scheme along with {\em any} centralized algorithm (instead of dynamic programming) for solving the problem inside the partition leads to performance comparable (or better) to that of the centralized algorithm for graph with polynomial growth. Then the resulting distributed algorithm becomes very fast in practice, unlike the dynamic programming approach, if the centralized algorithm inside the partition runs fast.
Similar guarantees can be obtained for minor-excluded graphs as well using the scheme utilized in \cite{nips08}. As mentioned earlier, the result is established for both MAP in pair-wise MRF and modularity optimization based clustering.

\paragraph{MAP Inference.}
Computing the exact Maximum a Posteriori (MAP) solution in a general probabilistic model is an NP-hard problem. A number of algorithmic approaches have been developed to obtain approximate solutions for these problems. Most of these methods work by making `local updates' to the assignment of the variables. Starting from an initial solution, the algorithms search the space of all possible local changes that can be made to the current solution (also called move space), and choose the best amongst them.

One such algorithm (which has been rediscovered multiple times) is called Iterated Conditional Modes or ICM for short. Its local update involves selecting (randomly or deterministically) a variable of the problem. Keeping the values of all other variables fixed, the value of the selected variable is chosen which results in a solution with the maximum probability. This process is repeated by selecting other variables until the probability cannot be increased further. The local step of the algorithm can be seen as performing inference in the smallest decomposed subgraph possible.

Another family of methods are related to max-product belief propagation (cf. \cite{Pea88} and \cite{YFW00}). In recent years a sequence
of results suggest that there is an intimate relation between the max-product algorithm and a natural linear programming relaxation -- for example, see \cite{WJWb, BSS05, BSS08, Jebara07, SSW07}. Many of these methods can be seen as making local updates to partitions of the dual problem \cite{SontagJ09,TarlowBKK11}.

We also note that Swendsen-Wang algorithm (SW)\cite{SW87}, a local flipping algorithm, has a philosophy similar to ours in
that it repeats a process of randomly partitioning the graph, and computing an assignment. However, the graph partitioning of SW
is fundamentally different from ours and there is no known guarantee for the error bound of SW.

In summary, all the approaches thus far with provable guarantees for local update based algorithm are primarily for linear or more generally convex optimization setup.

\paragraph{Modularity Optimization for Clustering.}
The notion of modularity optimization was introduced by Newmann \cite{newman} to identify
the communities or clusters in a network structure. Since then, it has become quite popular
as a metric to find communities or clusters in variety of networked data cf. \cite{B1, B2}. The major
challenge has been design of approximation algorithm for modularity optimization (which is
computationally hard in general) that can operate in distributed manner and provide performance guarantees.
Such algorithms with provable performance guarantees are known only for few cases, notably
logarithmic approximation of \cite{mod-opt-approx} via a centralized solution.

Our contribution in the context of modularity optimization lies in showing that indeed it is a
{\em decomposable} problem and therefore admits an distributed and fast
approximation algorithm through our approach.

\paragraph{Organization.}
The rest of the paper is organized as follows. Section \ref{sec:two}
	describes problem statement and preliminaries. Section \ref{sec:three} describes our
	main algorithms, and Section \ref{sec:four} presents analyses of our algorithms.
	Section \ref{sec:five2} and Section \ref{sec:five2} provides the proofs of our main theorems, and Section \ref{sec:six} presents the conclusion.

\section{Setup}\label{sec:two}

\noindent{\bf Graphs.} Our interest is in processing networked data represented through an undirected
graph $G = (V, E)$ with $n = |V|$ vertices and $E$ being the edge set. Let $m = |E|$ be the number of edges.
Graphs can be classified structurally
in many different ways: trees, planar, minor-excluded, geometric, expanding, and so on. We shall establish
results for graphs with geometric structure or polynomial growth which we define next. A graph
$G=(V,E)$ induces a natural `graph metric' on vertices $V$, denoted by $\bdg : V \times V \to \Rp$ with $\bdg(i, j)$
given by the length of the shortest path between $i$ and $j$; defined as $\infty$ if there is no path between
them.
\vspace{0.1cm}
\begin{definition}[Graph with Polynomial Growth]\label{def:polynomially}
We say that a graph $G$ (or a collection of graphs) has polynomial growth of degree  (or growth rate)
$\rho$, if for any $i \in V$ and $r \in {\mathbb N}$,
$$ | \mathbf{B}_G(i, r) | \leq C\cdot r^\rho,$$
where $C > 0$ is a universal constant and
~$\mathbf{B}_G(i,r)=\{j \in V| \mathbf{d}_G(i,j)< r\}.$
\end{definition}
Note that interesting values of $C, \rho$ are integral between $\{0, 1,\dots, n\}$, and it is easy to verify in $O(mn)$ time. Therefore we will assume knowledge of $C, \rho$ for algorithm design.  A large class of graph model naturally fall into the
graphs with polynomial growth. To begin with, the
standard $d$-dimensional regular grid graphs have
polynomial growth rate $d$. More generally, in recent years in the
context of computational geometry and metric embedding,
the graphs with finite doubling dimensions have become
popular object of study \cite{GKL031}. It can be checked
that a graph with doubling dimension $\rho$ is also
a graph with polynomial growth rate $\rho$. Finally,
the popular geometric graph model where nodes are
placed arbitrarily in some Euclidean space with some minimum distance separation, and two nodes have
an edge between them if they are within certain
finite distance, has finite polynomial growth rate \cite{GJSS091}.

\medskip

\noindent{\bf Pair-wise graphical model and MAP.} For a pair-wise Markov Random Filed (MRF) model defined on a graph $G = (V, E)$, each vertex
$i \in V$ is associated with a random variable $X_i$ which we
shall assume to be taking value from a finite alphabet $\Sigma$; the edge $(i, j) \in E$
represents a form of `dependence' between $X_i$ and $X_j$. More precisely, the joint distribution
is given by
\begin{align}
\mathbb P\big(\bX = \bx\big) & \propto \prod_{i \in V} \phi_i(x_i) \cdot \prod_{(i, j)\in
E} \psi_{ij}(x_i,x_j) ~~\label{markovicity}
\end{align}
where $\phi_i : \Sigma \to \Rp$ and $\psi_{ij}: \Sigma^2 \to \Rp$
are called node and edge potential functions\footnote{For simplicity of the analysis we assume strict positivity of $\phi_i$'s and
$\psi_{ij}$'s.}. The question of interest is to find the maximum a
posteriori (MAP) assignment $\bx^*\in \Sigma^n$, i.e.
$$ \bx^* \in \arg\max_{\bx \in \Sigma^n} \Pr[\bX = \bx].$$
Equivalently, from the optimization point of view, we wish to
find an optimal assignment of the problem
%\vspace{0.1cm}
$$ {\sf maximize} \quad \cH(\bx) \quad {\sf over} \quad \bx \in \Sigma^n, \quad \text{where}$$
$$\cH(\bx) = \sum_{i\in V} \ln \phi_i(x_i) + \sum_{(i, j)\in E} \ln \psi_{ij}(x_i, x_j).$$
For completeness and simplicity of exposition, we assume that
the function $\cH$ is finite valued over $\Sigma^n$. However, results
of this paper extend for hard constrained problems such
as the {\em hardcore} or {\em independent set} model.  We call
an algorithm $\alpha$ approximation for $\alpha \geq 1$ if it always produces assignment
$\widehat{\bx}$ such that
$$\frac{1}{\alpha} \cH(\bx^*) \leq \cH(\widehat{\bx}) \leq \cH(\bx^*).$$

\medskip

\noindent{\bf Social data and clustering/community.} Alternatively, in a social setting, vertices of graph $G$
can represents individuals and edges represent some form of interaction between them. For example,
consider a cultural show organized by students at a university with various acts. Let there be $n$ students
in total who have participated in one or more acts. Place an edge between two students if they participated
in at least one act together. Then the resulting graph represents interaction between students in terms of
acting together.

Based on this  observed network, the goal is to identify the set of all acts performed and
its `core' participants. The true answer, which classifies each student/node into the acts in which
s/he performed would lead to partitions of nodes in which a node may belong to multiple partitions.
Our interest is in identifying disjoint partitions which would, in this example, roughly mean
identification of `core' members of acts.

In general, to select a disjoint partition of $V$ given $G$, it is not clear what
is the appropriate criteria. Newman \cite{newman} proposed the notion of
modularity as a criteria. The intuition behind it is that a cluster or community
should be as distinct as possible from being `random'. Modularity of a partition of nodes is defined as the fraction of the edges
that fall within the  disjoint partitions minus the expected such fraction if edges were distributed at random with the same node degree sequences.
Formally, the modularity
of a subset $S \subset V$ is defined as
\begin{align}\label{eq:modularity}
M(S) & = \sum_{i, j \in S} \Big(A_{ij} - \frac{d_i d_j}{2m} \Big),
\end{align}
where $A_{ij} = 1$ iff $(i,j) \in E$ and $0$ otherwise, $d_i = |\{ k \in V: (i,k) \in E\}|$
is the degree of node $i \in V$, and $m = |E|$ represents the total number of edges in
$G$. More generally, the modularity of a partition of $V$,
$V = S_1 \cup \dots \cup S_{\ell}$ for some $1\leq \ell \leq n$ with
$S_{i} \cap S_j = \emptyset$ for $i \neq j$, is given by %defined as
\begin{align}\label{eq:modularity-1}
\cM(S_1,\dots, S_\ell) & = \frac{1}{2m} \Big(\sum_{i=1}^\ell M(S_i) \Big).
\end{align}
The modularity optimization approach \cite{newman} proposes to identify the community
structure as the disjoint partitions of $V$ that maximizes the total modularity, defined
as per \eqref{eq:modularity-1}, among all possible disjoint partitions of $V$ with
ties broken arbitrarily. The resulting clustering of nodes is the desired
answer.

We shall think of clustering as assigning {\em colors} to nodes. Specifically, given a
coloring $\chi : V \to \{1,\dots, n\}$, two nodes $i$ and $j$ are part of the same cluster
(partition) iff $\chi(i) = \chi(j)$. With this notation,  any clustering of $V$ can be represented
by some such coloring $\chi$ and vice versa.  Therefore, modularity optimization is
equivalent to finding a coloring $\chi$ such that its modularity $\cM(\chi)$ is maximized, where
\[
\cM(\chi) ~= \frac{1}{2m} \sum_{i, j \in V} \bone_{\{\chi(i) = \chi(j)\}} \Big(A_{ij} - \frac{d_i d_j}{2m}\Big).
\]
Here $\bone_{\{\cdot\}}$ is the indicator function with $\bone_{\{\mbox{\footnotesize true}\}} = 1$ and
$\bone_{\{\mbox{\footnotesize false}\}} = 0$. Let $\chi^*$ be a clustering that maximizes the modularity.
Then, as before, an algorithm will be said $\alpha$-approximate if it produces $\widehat{\chi}$
such that
\begin{align}
\frac{1}{\alpha} \cM(\chi^*) & \leq \cM(\widehat{\chi}) ~\leq \cM(\chi^*).
\end{align}

\section{Partition-Merge Algorithm}\label{sec:three}

We describe a parametric meta-algorithm for solving the MAP inference and modularity optimization.  The
meta-algorithm uses two parameters; a large constant $K \geq 1$ and a small real number $\varepsilon \in (0,1)$ to produce a
partition of $V = V_1\cup \dots \cup V_p$ so that each partition $V_j, 1\leq j\leq p$ is small. We will specify the values of $K$ and $\varepsilon$ in Section \ref{sec:four}. The
meta-algorithm uses an existing centralized algorithm to solve the original problem on each of
these partitioned sub-graphs $G_j = (V_j, E_j)$ independently where $E_j = (V_j \times V_j) \cap E$. The
result assignment leads to a candidate solution for the problem on entire graph. As we establish in Section \ref{sec:four},
this becomes a pretty good solution. Next, we describe the algorithm in detail.

\medskip
\noindent{\bf Step 1. Partition.}
We wish to create a partition of $V  = V_1\cup \dots \cup V_p$
for some $p$ with $V_i \cap V_j = \emptyset$ for $i \neq j$ so that the number of edges crossing
partitions are small. The algorithm for such partitioning is iterative. Initially, no node is part of
any partition. Order the $n$ nodes arbitrarily, say $i_1,\dots, i_n$. In iteration $k \leq n$,
choose node $i_k$ as the pivot. If $i_k$ belongs to $\cup_{\ell=1}^{k-1} V_\ell$, then
set $V_k = \emptyset$, and move to the next iteration if $k < n$ or else the algorithm concludes. If
$i_k \notin \cup_{\ell=1}^{k-1} V_\ell$, choose a radius $R_k \leq K$ at random with distribution
\begin{align}\label{eq:choiceofR}
\Pr\Big(R_k = \ell \Big) & = \begin{cases} \varepsilon (1-\varepsilon)^{\ell-1}  & \mbox{~for~} 1\leq \ell < K \\
                                                           (1-\varepsilon)^{K-1}, & \mbox{~for~} \ell = K.
\end{cases}
\end{align}
Let $V_k$ be set of all nodes in $V$ that are within distance $R_k$ of $i_k$, but that are not part of
$V_1\cup\dots\cup V_{k-1}$. Since we execute this step only if $i_k \notin \cup_{\ell=1}^{k-1} V_\ell$
and $R_k \geq 1$, $V_k$ will be non-empty. At the end of the $n$ iterations, we have a
partition of $V$ with at most $n$ non-empty partitions. Let the non-empty partitions of $V$
be denoted as $V = V_1\cup \dots \cup V_p$ for some $p \leq n$. A caricature of an iteration
is described in Figure 1.

%\begin{figure*}[t]
%\begin{center}
%\includegraphics[width=9cm]{graph.eps}
%\end{center}
%\caption{A pictorial description of an iteration of the graph partitioning.}
%\label{fig:decomp} %\vspace{-.1in}
%\end{figure*}

\begin{figure}
\centering
%\numberwithin{figure}{section}
\label{fig:decom} \includegraphics[width = 0.4 \textwidth]{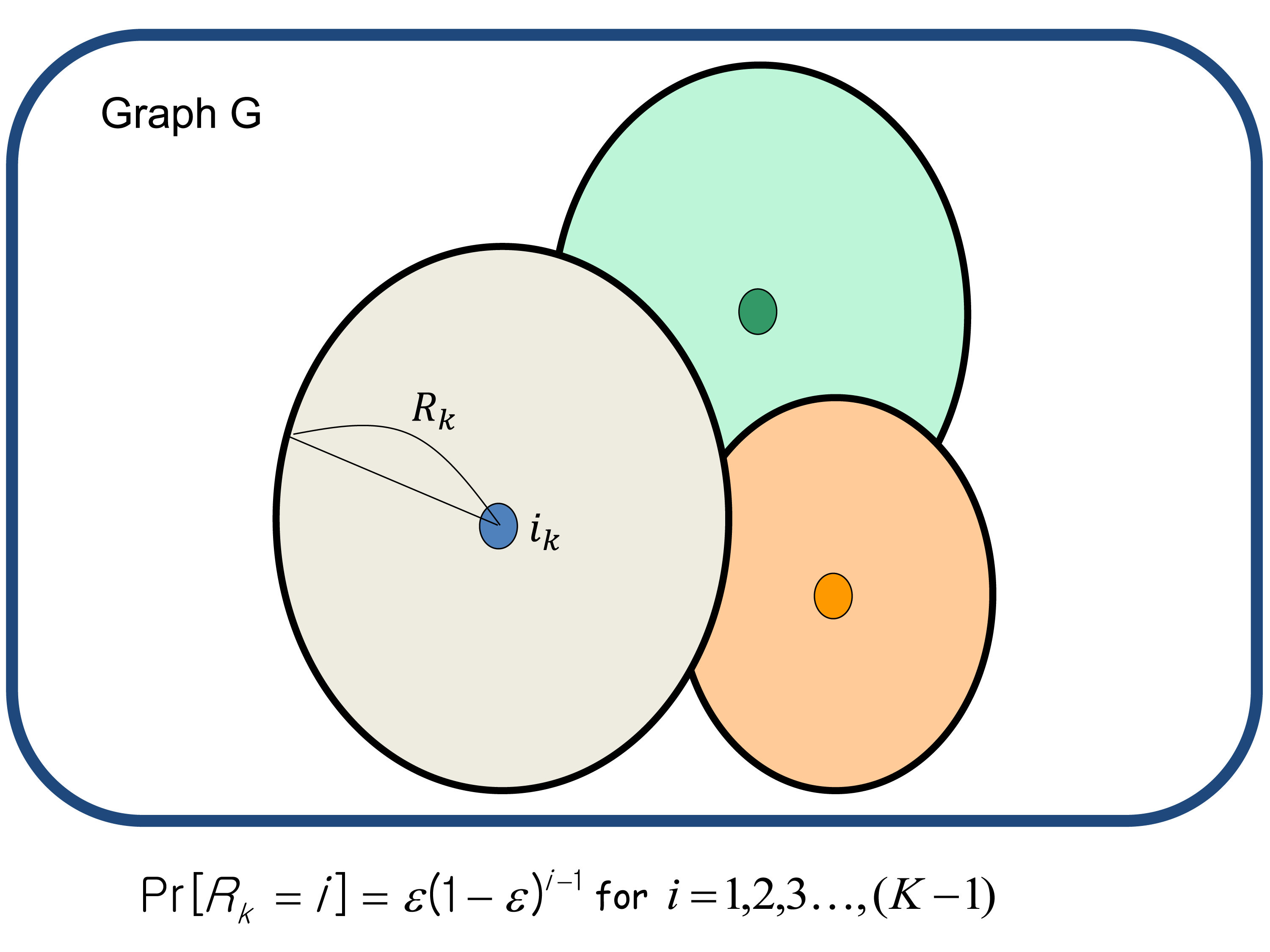}
\caption{A pictorial description of an iteration of the graph partitioning.}
\end{figure}

\medskip
\noindent{\bf Step 2. Merge (solving the problem).} Given the partition $V = V_1\cup \dots \cup V_p$,
consider the graphs $G_k = (V_k, E_k)$ with $E_k = (V_k \times V_k) \cap E$ for $1\leq k\leq p$. We shall
apply a centralized algorithm for each of these graph $G_1,\dots, G_k$ separately. Specifically,
let $\cA$ be an algorithm for MAP or for clustering: the algorithm may be exact (e.g. one solving problem
by exhaustive search over all possible options, or dynamic programming) or it may be an approximation
algorithm (e.g. $\alpha$-approximate for any graph). We apply $\cA$ for each subgraph
separately.
\begin{itemize}
\item[$\circ$] For MAP inference, this results in an assignment to all variables since in each partition each
node is assigned some value and collectively all nodes are covered by the partition. Declare thus
resulting global assignment, say $\widehat{\bx}$ as the solution for MAP.

\item[$\circ$] For modularity optimization, nodes in each partition $V_j$ are clustered. We declare the
union of all such clusters across partitions as the global clustering. Thus two nodes in different partitions
are always in different clusters; two nodes in the same partition are in different clusters if the centralized
algorithm applied to that partition clusters them differently.

\end{itemize}

\medskip
\noindent{\bf Computation cost.} The computation cost of the partitioning scheme scales linearly
in the number of edges in the graph. The computation cost of solving the problem in each of the components
$G_1,\dots, G_p$ depends on component sizes and on how the computation cost of algorithm $\cA$ scales
with the size. In particular, if the maximum degree of any node in $G$ is bounded, say by $d$,
then each partition has at most $d^K$ nodes. Then the overall cost is
$O(Q(d^K) n)$ where $Q(\ell)$ is the computation cost of $\cA$ for any graph with $\ell$ vertices.

\section{Main results}\label{sec:four}

\subsection{Graphs with polynomial growth}  We state sharp results for graphs with polynomial growth. We state results for
MAP inference and for modularity optimization under the same theorem statement to avoid repetition. The proofs, however, will
have some differences.

\begin{theorem}\label{thm:main-poly}
Let the graph $G = (V,E)$ have polynomial growth with degree $\rho \geq 1$ and constant $C \geq 1$. Then, for
a given $\delta \in (0,1)$, select parameters
\begin{align}
K = K(\rho, C, \delta) & = \frac{8\rho}{\varepsilon} \log\Big( \frac{8\rho}{\varepsilon}\Big) + \frac{4}{\varepsilon} \log C + \frac{4}{\varepsilon} \log \frac{1}{\varepsilon} + 2, \nonumber \\
\varepsilon = \varepsilon(\rho,C,\delta) & = \begin{cases} \frac{\delta}{2C2^\rho},  & \text{for~MAP} \\
              										\frac{\delta}{4(2C-1)}, & \text{for~modularity~optimization}.
										\end{cases}
\end{align}
Then, the following holds for the meta algorithm described in Section \ref{sec:three}.
\begin{itemize}
\item[(a)] If  $\cA$ solves the problem (MAP or modularity optimization) exactly, then the solution produced
by the algorithm $\widehat{\bx}$ and $\widehat{\chi}$ for MAP and modularity optimization respectively are such
that
\begin{align}\label{eq:1d}
(1-\delta) \cH(\bx^*) & \leq \E[\cH(\widehat{\bx})] \leq \cH(\bx^*) \nonumber \\
(1-\delta) \cM(\chi^*) & \leq \E[\cM(\widehat{\chi})] \leq \cM(\chi^*).
\end{align}

\item[(b)] If $\cA$ is $\alpha(n) \geq 1$ approximation algorithm for graphs with $n$ nodes, then
\begin{align}\label{eq:2d}
\left(\frac{1}{\alpha(\tilde{K})}-\delta\right) \cH(\bx^*) & \leq \E[\cH(\widehat{\bx})] \leq \cH(\bx^*), \nonumber \\
\frac{(1-\delta)}{\alpha(\tilde{K})}  \cM(\chi^*) & \leq \E[\cM(\widehat{\chi})] \leq \cM(\chi^*),
\end{align}
where $\tilde{K} = C K^\rho$.
\end{itemize}
\end{theorem}

%
%
%\begin{theorem}[\bf Modularity optimization]\label{thm:main-mod}
%Let the graph $G = (V,E)$ have polynomial growth with degree $\rho \geq 1$ and constant $C \geq 1$. Then, for
%a given $\delta \in (0,1)$, select parameters
%\begin{align}
%\varepsilon = \varepsilon(\rho,C,\delta) & = \frac{\delta}{4(2C-1)} \qquad  K = K(\rho, C, \delta) = \frac{8\rho}{\varepsilon} \log\Big( \frac{8\rho}{\varepsilon}\Big) + \frac{4}{\varepsilon} \log C + \frac{4}{\varepsilon} \log \frac{1}{\varepsilon} + 2.
%\end{align}
%Then, the following holds about the algorithm described above.
%\begin{itemize}
%%
%\item[(a)] If the procedure $\cA$ used finds exact solution to modularity optimization (e.g. by brute-force), then
%the proposed algorithm produces $\widehat{\chi}$ so that
%\[
%(1-\delta) \cM(\chi^*) \leq \E[\cM(\widehat{\chi})] \leq \cM(\chi^*).
%\]
%
%\item[(b)] If $\cA$ is $\alpha(n) \geq 1$ approximation algorithm for any graph with $n$ nodes, then
%\[
%\frac{(1-\delta)}{\alpha(\tilde{K})}  \cM(\chi^*) \leq \E[\cM(\widehat{\chi})] \leq \cM(\chi^*),
%\]
% where $\tilde{K} = C K^\rho$.
%
%%\item[(c)] The overall computation cost of the algorithm scales as $O(n Q(\rho, C, \delta))$
%%where $Q(\rho, C, \delta)$ depends on $\rho, C$ and $\delta$ and it is the bound on computation
%%cost of $\cA$ on a graph with at most $C K^\rho$ nodes.
%%
%\end{itemize}
%%
%\end{theorem}

\subsection{General graph}\label{sec:gen}

The theorem in the previous section was for graphs with polynomial growth. We now state results for general graph. Our result tells us how to evaluate the `error bound' on solutions produced by
the algorithm for any instantiation of randomness. The result is stated below for both MAP and modularity optimization. The
`error function' depends on the problem.

\begin{theorem}\label{thm:main-gen}
Given an arbitrary graph $G = (V, E)$ and our algorithm operating on it with parameters $K \geq 1$, $\varepsilon \in (0,1)$
using a known procedure $\cA$, the following holds:
\begin{itemize}
\item[(a)] If  $\cA$ solves the problem (MAP or modularity optimization) exactly, then the solution produced
by the algorithm $\widehat{\bx}$ and $\widehat{\chi}$ for MAP and modularity optimization respectively are such
that (with $\cB = E \backslash \cup_{k=1}^p E_k$),
\begin{align}\label{eq:1f}
\cH(\widehat{\bx})  & \geq~ \cH(\bx^*) -  \sum_{(i,j) \in \cB} \big({\psi_{ij}^U}-{\psi_{ij}^L}\big), \nonumber \\
\cM(\hat{\chi})  & \geq~ \cM(\chi^*) - \frac{|\cB|}{2m}.
\end{align}

\item[(b)] If $\cA$ is instead a $\alpha(n)$-approximation for graphs of size $n$, then
\begin{align}\label{eq:2f}
\cH(\widehat{\bx}) & \geq \frac{1}{\alpha({\tilde{K}})}\Big(\cH(\bx^*) - \sum_{(i,j) \in \cB} \big({\psi_{ij}^U}-{\psi_{ij}^L}\big)\Big) \nonumber \\
\cM(\tilde{\chi}) & \geq \frac{1}{\alpha({\tilde{K}})} \Big(\cM(\chi^*) - \frac{|\cB|}{2m }\Big),
\end{align}
where $\tilde{K}$ is the maximum number of nodes that are within $K$ hops of any single node in $V$.
\end{itemize}
In the expression above,
%where %$\cB$ is the (random) imaginary
$\psi_{ij}^U \triangleq\max_{\sigma, \sigma' \in \Sigma} \ln\psi_{ij}(\sigma,\sigma')$,
and $\psi_{ij}^L\triangleq\min_{\sigma, \sigma' \in \Sigma} \ln \psi_{ij}(\sigma,\sigma')$.
\end{theorem}

\subsection{Discussion of results}

Here we dissect implications of the above stated theorems. To start with, Theorem \ref{thm:main-poly}(a)
suggests that when graphs have polynomial growth, there exists a Randomized Polynomial
Time  Approximation Scheme (PTAS) for MAP computation and modularity optimization that has computation time scaling linearly with $n$.
%The result for MAP was derived in \cite{nips09}, but the result for modularity optimization is new.

The dependence on $\rho$ and  $\delta$ is rather stringent and therefore, even for moderately
small $\delta$, it  may not be possible to utilize existing computers to implement brute-force/dynamic programming
 procedure.  The Theorem \ref{thm:main-poly}(b) suggests that, if instead of using exact procedure for
each partition, when an approximation algorithm is used, the resulting solution almost retains its
approximation guarantees: if $\alpha(n)$ is a constant, then the resulting approximation
guarantee is essentially the same constant; if $\alpha(n)$ increases with $n$ (e.g. $\log n$), then
the resulting algorithm provides a constant factor approximation ! In either case, even if the approximation
algorithm has superlinear computation time in the number of nodes (e.g. semi-definite programming),
then our algorithm provides a way to achieve similar performance but in linear time for polynomially growing
graphs.

The algorithm, for general graph, produces a solution for which we have approximation guarantees.
Specifically, the error scales with the fraction of edges across partitions
that are induced by our partitioning procedure. This error depends on parameters
$K, \varepsilon$  utilized by our partitioning procedure. For graph with polynomial growth,
we provide recommendations on what the values should be for these parameters. However,
for general graph we do not have precise recommendations. Instead, one may try various
values of $K \in \{1,\dots, n\}$ and $\varepsilon \in (0,1)$ and then choose the best
solution. Indeed, a reasonable way to implement such procedure would be to take values
of $K$ that are $2^k$ for $k \in \{0,\dots, \log n\}$ and $\varepsilon$  chosen
at regular interval with granularity that an implementor is comfortable with (the smaller the
granularity, the better).

\section{Proofs of Theorems \ref{thm:main-poly}, \ref{thm:main-gen}: MAP inference}
\label{sec:five1}

In this Section, we first prove Theorem \ref{thm:main-poly}, and Theorem \ref{thm:main-gen} for MAP inference.

\vspace{0.15cm}
\noindent{\bf Bound on $|E \backslash \cup_{k=1}^p E_k|$.} We first state the following Lemma which shows the essential property of ~\MA. Lemma \ref{lem:three} will be used in the proofs of Theorems \ref{thm:main-poly}, \ref{thm:main-gen} both for MAP and modularity optimization.  The proof of Lemma \ref{lem:three} is stated at the end of this Section.
%We recall the following result from their work.
\begin{lemma}\label{lem:three}
Given $G = (V, E)$ with polynomial growth of rate $\rho \geq 1$ and associated
constant $C \geq 1$, by choosing $K = K(\rho, C, \delta)$ and $\varepsilon =\varepsilon(\rho, C,\delta)=\frac{\delta}{4 (2C-1)}$,
the partition scheme satisfies that for any edge $e\in E$,
\begin{align}
\mathbb P(e \in \cB ) \le 2\varepsilon.
\end{align}
\end{lemma}

\vspace{0.2cm}
\noindent{\bf Lower bound on $\cH(\bx^*)$.} Here we provide a lower bound on $\cH^*=\cH(\bx^*)$ that will be useful
to obtain multiplicative approximation property.
\begin{lemma}\label{lem:1one}
Let $\cH^* = \max_{\bx \in \Sigma^n} \cH(\bx)$ denote the maximum value of $\cH$ for a given pair-wise MRF on a graph $G$.
If $G$ has maximum vertex degree $d^*$, then
\beq
(d^*+1)\cH(\bx^*) \ge \sum_{(i,j) \in E} \lf({\psi_{ij}^U}-{\psi_{ij}^L} \rf).
\eeq
\end{lemma}

\begin{proof}
Assign weight $w_{ij} = \psi_{ij}^U$ to an edge $(i,j) \in E$.
Since graph $G$ has maximum vertex degree $d^*$, by Vizing's theorem there
exists an edge-coloring of the graph using at most $d^*+1$ colors.
Edges with the same color form a matching of the $G$. A standard
application of Pigeon-hole's principle implies that there is a color
with weight at least $\frac{1}{d^*+1} (\sum_{(i,j) \in E} w_{ij})$.
Let $M \subset E$ denote these set of edges. Then
$$ \sum_{(i,j) \in M} \psi_{ij}^U \geq \frac{1}{d^*+1}\lf(\sum_{(i,j) \in E} \psi_{ij}^U\rf).$$
Now, consider an assignment $\bx^M$ as follows: for each $(i,j) \in M$
set $(x_i^M, x_j^M) = \arg\max_{(x,x') \in \Sg^2} \psi_{ij}(x,x')$; for
remaining $i \in V$, set $x_i^M$ to some value in $\Sigma$ arbitrarily.
Note that for above assignment to be possible, we have used
matching property of $M$. Therefore, we have
\begin{eqnarray}
\cH(\bx^M) & = & \sum_{i \in V} \phi_i(x_i^M) + \sum_{(i,j) \in E} \psi_{ij}(x_i^M, x_j^M) \nonumber \\
& = & \sum_{i \in V} \phi_i(x_i^M) + \sum_{(i,j) \in E \backslash M} \psi_{ij}(x_i^M, x_j^M) + \sum_{(i,j) \in M} \psi_{ij}(x_i^M, x_j^M) \nonumber \\
& \stackrel{(a)}{\geq} & \sum_{(i,j) \in M} \psi_{ij}(x_i^M, x_j^M) \nonumber \\
& = & \sum_{(i,j) \in M} \psi_{ij}^U \nonumber \\
& \geq & \frac{1}{d^*+1} \lf[\sum_{(i,j)\in E} \psi_{ij}^U\rf].
\end{eqnarray}
Here (a) follows because $\psi_{ij}, \phi_i$ are non-negative valued
functions.  Since $\cH(\bx^*) \geq \cH(\bx^M)$ and $\psi_{ij}^L \geq
0$ for all $(i,j) \in E$, we prove Lemma \ref{lem:1one}.
\end{proof}

\medskip
\noindent{\bf Decomposition of $\cH^*$.} Here we show that by maximizing $\cH(\cdot)$
on a partition of $V$ separately and combining the assignments, the resulting $\widehat{\bx}$ has $\cH(\cdot)$ value as good as that
of MAP with penalty in terms of the edges across partitions.

\vspace{0.1in}

\begin{lemma}\label{lemma:diff2}
For a given MRF defined on $G$,
the algorithm \MA~produces output $\widehat{\bx}$ such that
$$ \cH(\widehat{\bx}) \ge \cH(\bx^*) - \left(\sum_{(i,j) \in \cB} \lf({\psi_{ij}^U}-{\psi_{ij}^L}\rf)\right),$$
where $\cB = E \backslash \cup_{k=1}^K E_k$,
$\psi_{ij}^U \triangleq\max_{\sigma, \sigma' \in \Sigma} \ln\psi_{ij}(\sigma,\sigma')$,
and $\psi_{ij}^L\triangleq\min_{\sigma, \sigma' \in \Sigma} \ln\psi_{ij}(\sigma,\sigma')$.
\end{lemma}

\begin{proof}
Let $\bx^*$ be a MAP assignment of the MRF $\bX$ defined on $G$. Given an assignment $\bx\in \Sigma^{|V|}$ defined on a graph $G=(V,E)$ and a subgraph $S=(W,E')$ of $G$, let an assignment $\bx'\in\Sigma^{|W|}$ be called {\em a restriction of $\bx$ to $S$} if
$\bx'(v)=\bx(v)$ for all $v\in W$.
Let $S_1,\ldots, S_K$ be the  connected components of $G'=(V,E-\cB)$, and let $\bx_k^*$ be the restriction of $\bx^*$ to the component $S_k$.
Let $\bX_k$ be the restriction of the MRF $\bX$ to $G_k=(S_k,E_k)$, where $E_k=\{(u,w)\in E | u,w\in S_k\}$.
% Then, since $\bx^*$ is a MAP assignment of $\bX$, $\bx_k^*$ is a MAP assignment of $\bX_k$.

For $\bx_k\in \Sigma^{|S_k|}$, define
$$\cH_k(\bx_k) = \sum_{i\in S_k} \phi_i(x_i) + \sum_{(i,j)\in E_k} \psi_{ij}(x_i,x_j).$$

Let $\widehat \bx$ be the output of ~\MA, and let $\widehat \bx_k$ be the restriction of $\widehat \bx$ to the component $S_k$.
Note that since $\widehat \bx_k$ is a MAP assignment of $\cH_k(\cdot)$ by the definition of our algorithm, for all $k=1,2,\ldots K$,

\beq \cH_k(\hat \bx_k) \ge \cH_k(\bx_k^*). \label{eq:ineq} \eeq

%$$\partial S_k=\{ u \in S_k| (u,w)\in \cB~ \mbox{for some} ~w\in V\}$$
%be the set of boundary vertices of $S_k$.
%Then the condition that all the vertices outside $S_k$ are fixed to an assignment is equivalent to changing the vertex potential functions on $\partial S_k$ due to the fixed end point edge potentials of $\cB_k$, and removing the edge potentials of $\cB_k$.

%Let $\bX_k$ be the MRF obtained from $\bX$ by restricting to $G_k$, then changing the boundary vertex potentials due to the fixed assignment outside $S_k$ at the iteration $T(S_k)$ of \MA.
%Let $\{\phi^k_i(\cdot)\}_{i\in S_k}$ be the vertex potential functions of $\bX_k$.
%For $i\in S_k$, let
%$$\phi_{i}^{k,D}=\max_{x_i\in \Sigma} \left|\phi_{i}(x_i)-\phi^k_{i}(x_i)\right|$$

%Then, from the potential change of $\phi^k_{i}$,  we have
%\beq\sum_i \phi_{i}^{D}\le \sum_{(u,w)\in \cB_j} (\psi_{ij}^U - \psi_{ij}^L).\label{eq:444}\eeq

Now, we have
\begin{eqnarray}
\cH(\hat \bx) - \cH(\bx^*)
&=& \sum_{k=1}^K \left[\cH_k (\widehat \bx_k) - \cH_k(\bx_k^*)\right] +\sum_{(i,j)\in \cB} \psi_{ij}(\widehat x_i,\widehat x_j) - \psi_{ij}(x^*_i, x^*_j)\nonumber\\
&\stackrel{(a)}\ge& \sum_{k=1}^K \left[\cH_k (\widehat \bx_k) - \cH_k(\bx_k^*)\right] - \sum_{(i,j)\in \cB} (\psi^U_{ij} - \psi^L_{ij})\nonumber\\
&\stackrel{(b)}\ge& - \sum_{(i,j)\in \cB} (\psi^U_{ij} - \psi^L_{ij})\label{eq:445}.
\end{eqnarray}
Here (a) follows from the definitions of $\psi^U_{ij}$ and $\psi^L_{ij}$, and (b) follows from (\ref{eq:ineq}).
This completes the proof of Lemma \ref{lemma:diff2}.
\end{proof}

\noindent{\bf Completing Proof of Theorem \ref{thm:main-poly}(a).}
Recall that the maximum vertex degree $d^*$ of $G$ is less than
$2^{\rho}C$ by the definition of polynomially growing graph. Remind our definition $\varepsilon=\frac{\delta}{2C 2^{\rho}}$ for MAP inference.
Now we have that
\begin{eqnarray}
\E[\cH(\widehat{\bx})] &\stackrel{(a)}\ge& \cH(\bx^*) - \E\left[\sum_{(i,j) \in \cB} \lf({\psi_{ij}^U}-{\psi_{ij}^L}\rf)\right] \\
&\stackrel{(b)}\ge& \cH(\bx^*) - 2\varepsilon \left(\sum_{(i,j) \in E} \lf({\psi_{ij}^U}-{\psi_{ij}^L}\rf)\right)\\
&\stackrel{(c)}\ge& \cH(\bx^*)\left( 1- 2\varepsilon (d^*+1)  \right)\\
&\stackrel{(d)}\ge& (1-\delta) \cH(\bx^*).
\end{eqnarray}
Here (a) follows from Lemma \ref{lemma:diff2}, (b) follows from Lemma \ref{lem:three}, (c) from Lemma \ref{lem:1one}, and (d) follows from the definition of $\varepsilon$ for MAP inference.
This completes the proof of Theorem \ref{thm:main-poly}(a) for MAP inference.

\medskip
\noindent{\bf Completing Proof of Theorem \ref{thm:main-poly}(b).} Suppose that we use an approximation procedure $\cA$
to produce an approximate MAP assignment $\widehat \bx_k$ on each partition $S_k$ in our algorithm. Let $\cA$ be such that the assignment produced
satisfies that $\cH_k(\widehat \bx_k)$ has value at least $1/\alpha(n)$ times the maximum $\cH_k(\cdot)$ value for any graph of size $n$. Now since
$\cA$ is applied to each partition separately, the approximation is within $\alpha(\tilde{K})$ where $\tilde{K} = C K^\rho$
is the bound on the number of nodes in each partition.
\begin{equation}
\cH(\widehat \bx_k) \geq \frac{1}{\alpha(\tilde{K})} \cM(\bx_k^*).\label{eq:approx}
\end{equation}

By the same proof of Lemma \ref{lemma:diff2} together with (\ref{eq:approx}), we have that

\begin{equation}
\E[\cH(\widehat{\bx})] \ge \frac{1}{\alpha(\tilde{K})}\cH(\bx^*) - \E\left[\sum_{(i,j) \in \cB} \lf({\psi_{ij}^U}-{\psi_{ij}^L}\rf)\right]. \label{eq:approx2}
\end{equation}
Hence we have that
\begin{eqnarray}
\E[\cH(\widehat{\bx})] &\ge& \frac{1}{\alpha(\tilde{K})}\cH(\bx^*) - \E\left[\sum_{(i,j) \in \cB} \lf({\psi_{ij}^U}-{\psi_{ij}^L}\rf)\right]\\
&\stackrel{(a)}\ge& \frac{1}{\alpha(\tilde{K})}\cH(\bx^*) - 2\varepsilon \left(\sum_{(i,j) \in E} \lf({\psi_{ij}^U}-{\psi_{ij}^L}\rf)\right)\\
&\stackrel{(b)}\ge& \cH(\bx^*)\left( \frac{1}{\alpha(\tilde{K})}- 2\varepsilon (d^*+1)  \right)\\
&\stackrel{(c)}\ge& \Big(\frac{1}{\alpha(\tilde{K})}-\delta\Big) \cH(\bx^*).
\end{eqnarray}

\noindent Here (a) follows from Lemma \ref{lem:three}, (b) follows from Lemma \ref{lem:1one}, and (c) from the definition of $\varepsilon$ for MAP inference.
This completes the proof of Theorem \ref{thm:main-poly}(b) for MAP inference.

\vspace{0.15cm}
\medskip
\noindent{\bf Completing Proof of Theorem \ref{thm:main-gen}.} The same arguments as in the proof Theorem \ref{thm:main-poly} together with Lemma \ref{lemma:diff2} completes the proof of Theorem \ref{thm:main-gen} for MAP inference.
%When $\cA$ produces exact solution to the
%modularity optimization for each partition, the resulting solution of our algorithm is $\hat{\chi}$.
%Therefore, the Lemma \ref{lem:two} immediately implies the desired claim. % of Theorem \ref{thm:main.1}(a).

%\medskip
%\noindent{\bf Completing Proof of Theorem \ref{thm:main-gen}(b).}

%Suppose we use an approximation procedure $\cA$
%to produce clustering on each partition in our algorithm. Let $\cA$ be such that the clustering produced
%has modularity at least $1/\alpha(n)$ times the optimal modularity for any graph of size $n$. Now since the
%$\cA$ is applied to each partition separately, the approximation is within $\alpha(\tilde{K})$ where $\tilde{K} = C K^\rho$
%is the bound on the number of nodes in each partition. Let $\tilde{\chi}^1,\dots,\tilde{\chi}^p$ be the clustering
%(coloring) produced by $\cA$ on graphs $G_1,\dots, G_p$. Then by property of $\cA$, we have
%\begin{align}
%\cM(\tilde{\chi}^k) & \geq \frac{1}{\alpha(\tilde{K})} \cM({\chi}^k).
%\end{align}
%Therefore, for the overall clustering $\tilde{\chi}$ obtained as union of $\tilde{\chi}^1,\dots,\tilde{\chi}^p$,
%\begin{align}
%\cM(\tilde{\chi}) & = \sum_{k=1}^p \cM(\tilde{\chi}^k)  %\nonumber \\
                         %&
%                         ~\geq \frac{1}{\alpha(\tilde{K})} \sum_{k=1}^p \cM({\chi}^k) %\nonumber \\
                         %&
%                         ~\geq \frac{1}{\alpha(\tilde{K})} \cM(\hat{\chi}).
%\end{align}
%From this and Theorem \ref{thm:main-gen}(a), the desired claim follows. %f Theorem \ref{thm:main.1}(b) follows immediately.

\vspace{0.25cm}
\noindent{\bf Proof of Lemma \ref{lem:three}.}\label{a:ld2}
Now we prove Lemma \ref{lem:three}. First, we consider property of \MA~applied to a generic metric space $\cG = (V, \bdg)$,
where $V$ is the set of points over which metric $\bdg$ is defined. We state the result below for any metric space (rather than
restricted to a graph) as it's necessary to carry out appropriate induction based proof. Note that the algorithm~\MA~can be
applied to any metric space (not just graph as well) as it only utilizes the property of metric in it's definition. The edge set
$E$ of metric space $\cG$ is precisely the set of all vertices that are within distance $1$ of each other.

\vspace{0.10in}
\begin{proposition}\label{claim:two}
Consider a metric space $\cG = (V, \bdg)$ defined over $n$ point set $V$, i.e. $|V|=n$.
Let $\cB = E \backslash \cup_{k=1}^p E_k$ be the boundary set of \MA~applied to $\cG$. Then, for any $e \in E$,
$$ \Pr[e \in \cB] \leq \varepsilon +  P_K \cdot | \bB(e, K)|, $$
where $\bB(e,K)=\bB_G(e,K)$ is the union of the two balls of radius $K$ in $\cG$ with respect to
the $\bdg$ centered around the two end vertices of $e$, and $P_K=(1-\varepsilon)^{K-1}$.
\end{proposition}
\vspace{0.10in}
\begin{proof}
The proof is by induction on the number of points $n$.   When $n=1$, the algorithm chooses only
point as $u_0$ in the initial iteration and hence no edge can be part
of the output set $\cB$. That is, for any edge, say $e$,
$$\Pr[e\in \cB] = 0 \leq \varepsilon + P_K |\bB(e,K)|.$$
Thus, we have verified the base case for induction ($n=1$).

As induction hypothesis, suppose that the Proposition \ref{claim:two} is
true for any graph with $n$ nodes with $n< N$ for some $N \geq
2$. As the induction step, we wish to establish Proposition \ref{claim:two} for any
$\cG=(V, \bdg)$ with $|V| = N$. For this, consider any $v \in V$. Now consider the last
iteration of the \MA~applied to $\cG$. The algorithm picks
$i_1 \in V$ uniformly at random in the first iteration. Given $e$,
depending on the choice of $i_1$ we consider three different cases
(or events). We will show that in these three cases,
$$ \Pr[e \in \cB] \leq \varepsilon +  P_K | \bB(e, K)|$$
holds.

\vspace{0.10in}

{\em Case 1.} Suppose $i_1$ is such that $\bdg(i_1, e) < K$, where the distance of a point and an edge of $\cG$
is defined as a minimum distance from the point to one of the two end-points of the edge. Call this event
$E_1$. Further, depending on choice of random number $R_1$, define the following events
$$ E_{11}= \{\bdg(i_1,e) < R_1\}, ~ E_{12}=\{\bdg(i_1,e) = R_1\}, ~\mbox{and}~ E_{13}=\{\bdg(i_1,e) > R_1\}.$$
By the definition of the partition scheme, when $E_{11}$ happens, $e$ can never be a part of $\cB$. When
$E_{12}$ happens, $e$ is
definitely a part of $\cB$. When $E_{13}$ happens, it is said to be
left as an element of the set $\cW_1$.  This new vertex set $\cW_1$ has
points less than $N$. The original metric $\bdg$ is still considered as the metric on
the points\footnote{Note the following subtle but crucial point. We are
not changing the metric $\bdg$ after we remove points from the original
set of points.} of $\cW_1$.  By its definition, \MA~ excluding the first iteration is the same as \MA~applied to $(\cW_1, \bdg)$. Therefore, we can invoke
induction hypothesis which implies that if event $E_{13}$ happens
then the probability of $v \in \cB$ is bounded above by $ \varepsilon +
P_K \cdot |\bB(e, K)|$, where $\bB(e,K)$ is the ball with respect to $(\cW_1, \bdg)$
which has no more than the number of points in the ball  $\bB(e,K)$
defined with respect to the original metric space $\cG$. Finally, let us relate the $\Pr[E_{11}|E_2]$ with
$\Pr[E_{12}|E_1]$. Suppose $\bdg(i_1, e) = \ell < K$. By the definition
of probability distribution of $\bQ$, we have \beq
 \Pr[E_{12}|E_1]   & = &  \varepsilon (1-\varepsilon)^{\ell-1},
 \eeq
 \beq
 \Pr[E_{11}|E_1]   & = &  (1-\varepsilon)^{K-1} + \sum_{j=\ell+1}^{K-1} \varepsilon (1-\varepsilon)^{j-1} \nonumber \\
  & = & (1-\varepsilon)^{\ell}.
  \eeq
That is,
 $$ \Pr[E_{12}|E_1] = \frac{\varepsilon}{1-\varepsilon} \Pr[E_{11}|E_1].$$
Let $q \stackrel{\triangle}{=} \Pr[E_{11}|E_1]$. Then, \beq \Pr[e
\in \cB | E_1] & = & \Pr[e \in \cB|E_{11} \cap E_1] \Pr[E_{11}|E_1]
+ \Pr[e\in
\cB|E_{12} \cap E_1]\Pr[E_{12}|E_1] \nonumber \\
& ~ & ~~+ \Pr[e\in \cB|E_{13} \cap E_1 ]\Pr[E_{13}|E_1] \nonumber \\
& \leq & 0 \times  q + 1 \times \frac{\varepsilon q}{1-\varepsilon}  + (\varepsilon + P_K |\bB(e,K)|) \lf(1-\frac{q}{1-\varepsilon}\rf) \nonumber \\
& = & \varepsilon + P_K |\bB(e,K)| + \frac{q}{1-\varepsilon}\left(\varepsilon - \varepsilon - P_K|\bB(e,K)|\right) \nonumber \\
& = & \varepsilon + P_K |\bB(e,K)| - \frac{q P_K|\bB(e,K)|}{1-\varepsilon} \nonumber \\
& \leq &   \varepsilon + P_K |\bB(e,K)|. \eeq

{\em Case 2.} Now, suppose $i_1\in V$ is such that $\bdg(i_1, e) =
K$. We will call this event $E_2$.  Further, define the event $E_{21}=\{R_1 = K\}$.
Due to the independence of selection of $R_1$,  $\Pr[E_{21} | E_2] =
P_K$. Under the event $E_{21} \cap E_2$, $e \in \cB$ with probability
$1$. Therefore, \beq
\Pr[e \in \cB | E_2] & = & \Pr[e \in \cB|E_{21} \cap E_2] \Pr[E_{21}|E_2] + \Pr[e \in \cB|E^c_{21} \cap E_2] \Pr[E^c_{21}|E_2] \nonumber \\
& = & 1 \times P_K + \Pr[e \in \cB | E_{21}^c \cap E_2] (1-P_K).
\label{xx11} \eeq Under the event $E^c_{21} \cap E_2$, we have $e \in
\cW_1$, and the remaining metric space $(\cW_1, \bdg)$. This metric
space has $< N$ points. Further, the ball of radius $K$ around $e$
with respect to this new metric space has at most $|\bB(e, K)| - 1$
points (this ball is with respect to the original metric space $\cG$
on $N$ points). Now we can invoke the induction hypothesis for this new
metric space  to obtain \beq \Pr[e \in \cB | E^c_{21} \cap E_2] & \leq &
\varepsilon + P_K \cdot(|\bB(e,K)| - 1). \label{xx22} \eeq From (\ref{xx11}) and
(\ref{xx22}), we have
\beq
\Pr[e \in \cB | E_3] & \leq & P_K + (1-P_K) (\varepsilon + P_K \cdot(|\bB(e,K)| - 1)) \nonumber \\
& = &  \varepsilon (1-P_K) + P_K |\bB(e,K)| + P_K^2 (1- |\bB(e,K)|) \nonumber \\
& \leq & \varepsilon + P_K |\bB(e,K)|.\nonumber \eeq
In above, we have used the fact that $|\bB(e,K)| \geq 1$ (or else, the bound was trivial to begin with).

\vspace{0.10in}

{\em Case 3.} Finally, let $E_3$ be the event that $\bdg(i_1, e) >
K$. Then, at the end of the first iteration of the algorithm, we
again have the remaining metric space $(\cW_1, \bdg)$ such that
$|\cW_1| < N$. Hence, as before, by induction hypothesis we have
$$\Pr[e\in \cB|E_3] \le \varepsilon + P_K |\bB (e,K)|.$$

Now, the three cases are exhaustive and disjoint. That is,
$\cup_{i=1}^3 E_i$ is the universe. Based on the above discussion,
we obtain the following. \beq
\Pr[e \in \cB] & = & \sum_{i=1}^3 \Pr[ e \in \cB | E_i] \Pr[E_i] \nonumber \\
 & \leq & \lf(\max_{i=1}^3 \Pr[ e \in \cB | E_i]\rf) \lf(\sum_{i=1}^3 \Pr[E_i] \rf) \nonumber \\
 & \leq & \varepsilon + P_K \cdot|\bB(e,K)|.
 \eeq
This completes the proof of Proposition \ref{claim:two}.
\end{proof}
\noindent Now, we will use Proposition \ref{claim:two} to complete the proof of Lemma \ref{lem:three}.
The definition of growth rate implies that,
$$ \lf|\bB(e, K)\rf| \leq  C \cdot K^\rho.$$
From the definition  $P_K=(1-\varepsilon)^{K-1}$, we have
$$ P_K \lf|\bB(e, K)\rf| \le C(1-\varepsilon)^{K-1} K^\rho.$$
Therefore, to show Lemma \ref{lem:three}, it is sufficient to show that our definition of $K$ satisfies the following Lemma.

\vspace{0.1in}
\begin{lemma}\label{lemma:new} We have that
$$C(1-\varepsilon)^{K-1} K^\rho \leq \varepsilon.$$
\end{lemma}
\vspace{0.1in}
\begin{proof}

We will show the following equivalent inequality.
\beq (K-1)\log (1-\varepsilon)^{-1} \ge \rho \log K +\log C + \log \frac{1}{\varepsilon}.\label{eq:032}\eeq
First, note that for all $\varepsilon\in (0,1),$ $$\log (1-\varepsilon)^{-1}\ge \log(1+\varepsilon) \ge \frac{\varepsilon}{2}.$$
Hence to prove (\ref{eq:032}), it is sufficient to show that
\beq K \ge \frac{2\rho}{\varepsilon} \log K + \frac{2}{\varepsilon}\log C + \frac{2}{\varepsilon}\log \frac{1}{\varepsilon} +1.\label{eq:306}\eeq
Recall that $$ K= K(\varepsilon, \rho) = \frac{8\rho}{\varepsilon} \log
\lf(\frac{8\rho}{\varepsilon}\rf) + \frac{4}{\varepsilon}\log C + \frac{4}{\varepsilon}\log \frac{1}{\varepsilon} +2.$$
From the definition of $K$, we will show that
$$\frac{K}{2} \ge \frac{2\rho}{\varepsilon}\log K$$
and $$\frac{K}{2}\ge \frac{2}{\varepsilon}\log C + \frac{2}{\varepsilon}\log \frac{1}{\varepsilon} +1,$$ which will prove (\ref{eq:306}). The
following is straightforward:
\beq \frac{K}{2}\ge \frac{2}{\varepsilon}\log C + \frac{2}{\varepsilon}\log \frac{1}{\varepsilon} +1.\label{eq:401} \eeq
Now, let $\hat K = \frac{8\rho}{\varepsilon} \log
\lf(\frac{8\rho}{\varepsilon}\rf).$ Then
$$\frac{\hat K}{2}
= \frac{4\rho}{\varepsilon}\log \lf(\frac{8\rho}{\varepsilon}\rf) \ge
\frac{2\rho}{\varepsilon}\lf( \log \lf(\frac{8\rho}{\varepsilon}\rf) + \log
\log \lf(\frac{8\rho}{\varepsilon}\rf)\rf) = \frac{2\rho}{\varepsilon}\log
\hat K.$$
That is, $\frac{\hat K}{2} - \frac{2\rho}{\varepsilon}\log
\hat K \ge 0$. Since the function $\phi (x) = \frac{x}{2} - \frac{2\rho}{\varepsilon}\log x ~$ is an increasing function of $x$ when $x\ge \frac{4\rho}{\varepsilon}$, and from the fact that $K \ge \hat K \ge \frac{4\rho}{\varepsilon},$ we have
\beq \frac{K}{2} \ge \frac{2\rho}{\varepsilon}\log K.\label{eq:402}\eeq
From (\ref{eq:401}) and (\ref{eq:402}), we have (\ref{eq:306}), which completes the proof of Lemma  \ref{lemma:new}.
\end{proof}

\section{Proofs of Theorems \ref{thm:main-poly}, \ref{thm:main-gen}: Modularity optimization}
\label{sec:five2}

In this Section, we prove Theorem \ref{thm:main-poly}, and Theorem \ref{thm:main-gen} for modularity optimization.

\vspace{0.15cm}

\noindent{\bf Lower bound on $\cM^*$.} Here we provide a lower bound on $\cM^*$ that will be useful
to obtain multiplicative approximation property.
\begin{lemma}\label{lem:one}
Let $\cM^* = \max_{\chi} \cM(\chi)$ denote the maximum value of modularity for graph $G$. Then,
\[
\cM^* \geq \frac{1}{2(2C-1)} ~\Big(1-\frac{C^2}{2m}\Big).
\]
\end{lemma}
\begin{proof}
Since the graph has polynomial growth with degree $\rho$ and associated constant $C$, it
follows that the number of nodes within one hop of any node $i \in V$ (i.e. its immediate
neighbors) is at most $C$. That is, $d_i \leq C$ for all $i \in V$. Given this bound, it
follows that there exists a matching of size at least $m/(2C-1)$ in $G$. Given such a matching,
consider the following clustering (coloring). Each edge in the matching represent a community
of size $2$, while all the nodes that are unmatched lead to community of size $1$. By
definition, the individual (unmatched) nodes contribute $0$ to the modularity. The nodes
that are part of the two node communities, each contribute at least
$\frac{1}{2m}\big(1-\frac{C^2}{2m}\big)$ since vertex degree of each node is bounded
above by $C$. Since there are $m/(2C-1)$ edges in the matching, it follows that the net
modularity of such community assignment is at least $\frac{1}{2(2C-1)}\big(1-\frac{C^2}{2m}\big)$.
This completes the proof of Lemma \ref{lem:one} (Similar result, with tighter constant, follows from \cite{Han}).
\end{proof}

\medskip
\noindent{\bf Decomposition of $\cM^*$.} Here we show that by maximizing modularity
on a partition of $V$ separately, the resulting clustering has modularity as good as that
of optimal partitioning with penalty in terms of the edges across partitions. To that end,
let $V = V_1 \cup \dots \cup V_p$  be a partition of $V$, i.e. $V_i \cap V_j = \emptyset$
for $i \neq j$. Let $G_k = (V_k, E_k)$, where $E_k = (V_k \times V_k)\cap E$, denote the subgraph
of $G$ for $1\leq k\leq p$. Let $\chi^k$ be a coloring (clustering) of $G_k$ with maximum
modularity. Let $\chi^*$ be a coloring of $G$ with maximum modularity ($\cM^*$) and
let $\chi^{*,k}$ be the restriction of $\chi^*$ to $G_k$. Let $\hat{\chi}$ denote the
clustering of $G$ obtained by taking union of clusterings $\chi^1,\dots, \chi^p$.
Then we claim the following.
\begin{lemma}\label{lem:two}
For any partition $V = V_1\cup \dots \cup V_p$,
\[
\cM(\hat{\chi}) ~\geq \cM(\chi^*) - \frac{1}{2m}| E \backslash \cup_{k=1}^p E_k |.
\]
\end{lemma}
\begin{proof}
Consider the following:
\begin{align}
2m ~\cM(\hat{\chi}) & = \sum_{i, j \in V} \bone_{{\small \{\hat{\chi}(i) = \hat{\chi}(j)\}}} \Big(A_{ij} - \frac{d_i d_j}{2m}\Big) %\nonumber \\
            ~\stackrel{(a)}{=} \sum_{k=1}^p \sum_{i,j \in V_k} \bone_{{\small \{\chi^k(i) = \chi^k(j)\}}} \Big(A_{ij} - \frac{d_i d_j}{2m}\Big) \nonumber \\
           & \stackrel{(b)}{\geq} \sum_{k=1}^p \sum_{i, j \in V_k} \bone_{{\small \{\chi^{*,k}(i) = \chi^{*,k}(j)\}}} \Big(A_{ij} - \frac{d_i d_j}{2m}\Big) \nonumber \\
           & = \sum_{i, j \in V} \bone_{{\small \{\chi^{*}(i) = \chi^{*}(j)\}}} \Big(A_{ij} - \frac{d_i d_j}{2m}\Big)% \nonumber \\
%& \quad
- \sum_{{\small (i, j) \in V^2 \backslash \cup_{k=1}^p V_k^2}}  \bone_{{\small \{\chi^{*}(i) = \chi^{*}(j)\}}} \Big(A_{ij} - \frac{d_i d_j}{2m}\Big) \label{eq:one-x} \\
& \geq \sum_{i, j \in V} \bone_{{\small \{\chi^{*}(i) = \chi^{*}(j)\}}} \Big(A_{ij} - \frac{d_i d_j}{2m}\Big) - |E \backslash  \cup_{k=1}^p E_k|, \label{eq:one}
\end{align}
where the last inequality follows because the term inside the summation in \eqref{eq:one-x} is positive
only if $A_{ij} =1$, i.e. $(i,j) \in E$ or else it is negative. Therefore, for the purpose of lower bound, we only
need to worry about $(i,j) \in E$ such that $(i,j) \notin \cup_{k=1}^p V_k \times V_k$. This is precisely equal to
$E \backslash \cup_{k=1}^p E_k$. The (a) follows because $\hat{\chi}$, by definition, assigns nodes in
$V^i$ and $V^j$ for $i \neq j$ to different clusters. The (b)
follows because $\chi^k$ has maximum modularity in $G_k$ and hence it is at least as large (in terms of
modularity) as that of the $\chi^{*,k}$, the restriction of $\chi^*$ to $G_k$. This completes the
proof of Lemma \ref{lem:two} since the first term in \eqref{eq:one} is precisely $2m \cM(\chi^*) = 2m \cM^*$.
\end{proof}

%\medskip
\noindent{\bf Approximation factor for $\cM(\hat{\chi})$.} Let $\beta = |E \backslash \cup_{k=1}^p E_k|/m$ denote
the fraction of edges that are across partitions for a given partition $V= V_1 \cup\dots\cup V_p$. Then, from
Lemmas \ref{lem:one} and \ref{lem:two}, it follows that for $m \geq C^2$,
\begin{align}\label{eq:lemthree}
\cM(\hat{\chi}) & \geq \cM(\chi^*) \Big(1 - \frac{\beta}{2\cM(\chi^*)}\Big) ~ \geq \cM(\chi^*) \Big(1-2 (2C-1) \beta\Big).
\end{align}
Therefore, if $2 (2C-1)\beta \leq \delta$, then $\cM(\hat{\chi})$ is at least $\cM^*\cdot(1-\delta)$.
Now from Lemma \ref{lem:three} and the linearity of expectation, we have
\begin{align}\label{eq:lemthree.1}
\E\big[|E \backslash \cup_{k=1}^p E_k|\big] & \leq \frac{\delta}{2 (2C-1)} m.
\end{align}

%\medskip
\vspace{0.15cm}
\noindent{\bf Completing Proof of Theorem \ref{thm:main-poly}(a).} When $\cA$ produces exact solution to the
modularity optimization for each partition, the resulting solution of our algorithm is $\hat{\chi}$.
Therefore, from \eqref{eq:lemthree} and \eqref{eq:lemthree.1}, it follows that
\begin{align}
\E[\cM(\hat{\chi})] & \geq \cM(\chi^*) (1-\delta).
\end{align}

\medskip
\noindent{\bf Completing Proof of Theorem \ref{thm:main-poly}(b).} Suppose we use an approximation procedure $\cA$
to produce clustering on each partition in our algorithm. Let $\cA$ be such that the clustering produced
has modularity at least $1/\alpha(n)$ times the optimal modularity for any graph of size $n$. Now since
$\cA$ is applied to each partition separately, the approximation is within $\alpha(\tilde{K})$ where $\tilde{K} = C K^\rho$
is the bound on the number of nodes in each partition. Let $\tilde{\chi}^1,\dots,\tilde{\chi}^p$ be the clustering
(coloring) produced by $\cA$ on graphs $G_1,\dots, G_p$. Then by the approximation property of $\cA$, we have
\begin{align}
\cM(\tilde{\chi}^k) & \geq \frac{1}{\alpha(\tilde{K})} \cM({\chi}^k).
\end{align}
Therefore, for the overall clustering $\tilde{\chi}$ obtained as union of $\tilde{\chi}^1,\dots,\tilde{\chi}^p$, we have
\begin{align}\label{eq:cm}
\cM(\tilde{\chi}) & = \sum_{k=1}^p \cM(\tilde{\chi}^k)  %\nonumber \\
                         %&
                         ~\geq \frac{1}{\alpha(\tilde{K})} \sum_{k=1}^p \cM({\chi}^k) %\nonumber \\
                         %&
                         ~= \frac{1}{\alpha(\tilde{K})} \cM(\hat{\chi}).
\end{align}
Since $\E[\cM(\hat{\chi})]$ is at least $(1-\delta)  \cM^*$, it follows that
$\E[\cM(\tilde{\chi})]\ge \frac{(1-\delta)}{\alpha(\tilde{K})}\cM^*$.

\medskip
\noindent{\bf Completing Proof of Theorem \ref{thm:main-gen}.} 
Lemma \ref{lem:two} directly proves Theorem \ref{thm:main-gen}(a), and the same arguments as in the proof Theorem \ref{thm:main-poly}(b) completes the proof of Theorem \ref{thm:main-gen}(b).

%When $\cA$ produces exact solution to the
%modularity optimization for each partition, the resulting solution of our algorithm is $\hat{\chi}$.
%Therefore, the Lemma \ref{lem:two} immediately implies the desired claim. % of Theorem \ref{thm:main.1}(a).

%\medskip
%\noindent{\bf Completing Proof of Theorem \ref{thm:main-gen}(b).}

%Suppose we use an approximation procedure $\cA$
%to produce clustering on each partition in our algorithm. Let $\cA$ be such that the clustering produced
%has modularity at least $1/\alpha(n)$ times the optimal modularity for any graph of size $n$. Now since the
%$\cA$ is applied to each partition separately, the approximation is within $\alpha(\tilde{K})$ where $\tilde{K} = C K^\rho$
%is the bound on the number of nodes in each partition. Let $\tilde{\chi}^1,\dots,\tilde{\chi}^p$ be the clustering
%(coloring) produced by $\cA$ on graphs $G_1,\dots, G_p$. Then by property of $\cA$, we have
%\begin{align}
%\cM(\tilde{\chi}^k) & \geq \frac{1}{\alpha(\tilde{K})} \cM({\chi}^k).
%\end{align}
%Therefore, for the overall clustering $\tilde{\chi}$ obtained as union of $\tilde{\chi}^1,\dots,\tilde{\chi}^p$,
%\begin{align}
%\cM(\tilde{\chi}) & = \sum_{k=1}^p \cM(\tilde{\chi}^k)  %\nonumber \\
                         %&
%                         ~\geq \frac{1}{\alpha(\tilde{K})} \sum_{k=1}^p \cM({\chi}^k) %\nonumber \\
                         %&
%                         ~\geq \frac{1}{\alpha(\tilde{K})} \cM(\hat{\chi}).
%\end{align}
%From this and Theorem \ref{thm:main-gen}(a), the desired claim follows. %f Theorem \ref{thm:main.1}(b) follows immediately.

\section{Conclusion}\label{sec:six}

In recent years, it has become increasingly important to design distributed high-performance graph computation algorithms that can deal with large-scale networked data in a cloud-like distributed computation architecture. Inspired by this, in this paper, we have introduced Partition-Merge, a simple meta-algorithm, that takes an existing centralized algorithm and produces a distributed implementation. The resulting distributed implementation, with the underlying graph having polynomial growth property, runs in essentially linear time and is as good as, and sometimes even better than the centralized algorithm.

The algorithm is applicable to any graph in general, and its computation time as well as performance guarantees depend on the underlying graph structure -- interestingly enough, we have evaluated the performance guarantees for any graph. %Similar results (with partitioning scheme of \cite{nips08}) will hold for minor-excluded graphs as well (Planar graphs are special case of minor-excluded graph family). 
We strongly believe that such an algorithmic approach would be of great value for developing large-scale cloud-based graph computation facilities.

\acks{Part of this work appeared in the preliminary version \cite{nips09}. This work is supported in parts by Army Research Office under MURI Award 58153-MA-MUR, and in part by Basic Science Research Program through the National Research Foundation of Korea(NRF) funded by the Ministry of Education, Science and Technology(2012032786).}

\bibliographystyle{plain}
\bibliography{biblio}

\end{document}